\def\year{2018}\relax
\documentclass[letterpaper]{article} 
\usepackage{aaai18}  
\usepackage{times}  
\usepackage{helvet}  
\usepackage{courier}  
\usepackage{url}  
\usepackage{graphicx}  
\usepackage[hang,small,bf]{caption}
\usepackage[subrefformat=parens]{subcaption}
\usepackage{color}
\usepackage{amsmath,amssymb,amsthm,mathtools,bm,multirow}
\usepackage[linesnumbered,ruled,vlined]{algorithm2e}
\usepackage{comment}
\usepackage{tikz,pgfplots,pxpgfmark}
\usetikzlibrary{calc,shapes.callouts,decorations.text,shadows}

\newcommand{\ot}{\leftarrow}
\newcommand{\argmax}{\mathop{\rm argmax}}
\newcommand{\argmin}{\mathop{\rm argmin}}

\newcommand{\Ch}{\mathop{\mathrm{Ch}}\nolimits}

\newcommand{\smax}{\overline{s}}
\renewcommand{\mid}{\,:\,}
\usepackage{pifont}
\newtheorem{theorem}{Theorem}
\newtheorem{lemma}{Lemma}
\newtheorem{proposition}{Proposition}

\newtheorem{definition}{Definition}
\newtheorem{example}{Example}

\begin{document}
%
\title{Approximately Stable Matchings with Budget Constraints\thanks{
    A full version 
    can be found at http://arxiv.org/abs/1711.07359. 
    Supported in part by JST ACT-I and KAKENHI 26280081, 16K16005, and 17H01787.}}
\author{{\bf Yasushi Kawase}\\
  {\normalsize Tokyo Institute of Technology and} \\
  {\normalsize RIKEN AIP Center, Tokyo, Japan.}\\
  {\normalsize \texttt{kawase.y.ab@m.titech.ac.jp}}
  \And 
  {\bf Atsushi Iwasaki}\\
  {\normalsize University of Electro-Communications and} \\
  {\normalsize RIKEN AIP Center, Tokyo, Japan.}\\
  {\normalsize \texttt{iwasaki@is.uec.ac.jp}}
}
\maketitle
\begin{abstract}
  This paper examines two-sided matching with budget constraints 
  where one side (a firm or hospital) can make monetary transfers (offer wages) to the other (a worker or doctor). 
  In a standard model, while multiple doctors can be matched to a single hospital, 
  a hospital has a {\em maximum quota}; thus, the number of doctors assigned to a hospital cannot exceed a certain limit. 
  In our model, in contrast, a hospital has a {\em fixed budget}; that is, 
  the total amount of wages allocated by each hospital to doctors is constrained. 
  With budget constraints, stable matchings may fail to exist and checking for the existence is hard.
  To deal with the nonexistence, 
  we extend the ``matching with contracts'' model of Hatfield and Milgrom
  so that it deals with \textit{approximately stable} matchings where 
  each of the hospitals' utilities after deviation can increase by a factor up to a certain amount. 
  We then propose two novel mechanisms that efficiently return a stable matching that exactly satisfies the budget constraints. 
  Specifically, by sacrificing strategy-proofness, our first mechanism achieves the best possible bound.
  We also explore a special case on which a simple mechanism is strategy-proof for doctors,
  while maintaining the best possible bound of the general case. 
\end{abstract}


\section{Introduction}

This paper studies a two-sided, one-to-many matching model
when there are budget constraints on one side (a firm or hospital), that is, 
the total amount of wages that it can pay to the other side (a worker or doctor) is limited.
The theory of two-sided matching has been extensively developed, as illustrated by the comprehensive surveys of
\citeauthor{Roth:CUP:1990}~\shortcite{Roth:CUP:1990} or
\citeauthor{manlove:2013}~\shortcite{manlove:2013}.
Rather than fixed budgets, maximum quotas limiting the total number of doctors that each hospital can hire are typically used. 
 
Some real-world examples are subject to matching with such budget constraints: 
a college can offer stipends to recruit better students when the budget for admission is limited; 
a firm can offer wages to workers under the condition that employment costs depend on earnings in the previous accounting period; 
a public hospital can offer salaries to doctors when the total amount relies on funds from the government; and so on. 
To establish our model and concepts, 
we use doctor-hospital matching as a running example. 

To date, most papers on matching with monetary transfers assume that budgets are unrestricted,
e.g., \citeauthor{kelso:ecma:1982}~\shortcite{kelso:ecma:1982}. 
When they are restricted, stable matchings may fail to exist~\cite{mongell:el:1986,abizada:te:2016}. 
%
There are several other possibilities to circumvent the nonexistence problem. 
\citeauthor{abizada:te:2016}~\shortcite{abizada:te:2016} modifies the notion of stability in a different way 
and proposes a variant of the Deferred Acceptance (DA) mechanism 
that produces a {\em pairwise}, instead of coalitional, stable matching and is strategy-proof for doctors. 
\citeauthor{dean:ifip:2006}~\shortcite{dean:ifip:2006} assume that hospitals' priorities
are lexicographic to ensure the existence of a stable matching.
We instead allow each hospital to have an \textit{additive} utility, but the existence is not guaranteed yet. 
\citeauthor{kawase:ijcai:2017}~\shortcite{kawase:ijcai:2017} focus on
\textit{near-feasible} matchings that exceed each budget of the hospitals by a certain amount.
Their mechanisms find a ``nearby'' instance with a stable matching for each instance of a matching problem.

This paper focuses on \textit{approximately stable} matchings where
the participants are willing to change the assignments only for a multiplicative improvement of a certain amount~\cite{arkin2009}.
This idea can be interpreted as one in which
a hospital in a blocking pair changes his match as soon as his utility after the change increases by any (arbitrarily small) amount.
\citeauthor{arkin2009}~\shortcite{arkin2009} examine a stable roommate problem, which is a non-bipartite one-to-one matching problem, 
while we examine a bipartite one-to-many matching problem.
It is reasonable for a hospital to change his assignment only in favor of a significant improvement 
(the grass may be greener on the other side, but it takes effort to cross the fence).
  
Furthermore, it must be emphasized that these studies, except the works of 
\citeauthor{abizada:te:2016}~\shortcite{abizada:te:2016} and \citeauthor{kawase:ijcai:2017}~\shortcite{kawase:ijcai:2017}, 
discuss no strategic issue, that is, misreporting a doctor's preference may be profitable. 
The literature on matching has found strategy-proofness for doctors, i.e., 
no doctor has an incentive to misreport her preference, 
to be a key property in a wide variety of settings~\cite{Abdulkadiroglu:AER:2003}.

The contribution of this paper is twofold: 
First, we modify the generalized DA algorithm and devise a new property, 
which we call \textit{$\alpha$-approximation}, on the \textit{matching with contract} framework~\cite{Hatfield:AER:2005}.
This is because the existing class of mechanisms and properties are not sufficient
to characterize a choice function that produces such an approximately stable matching.
Second, we propose two novel mechanisms that efficiently return a stable matching that exactly satisfies the budget constraints. 
Specifically, by sacrificing strategy-proofness, the best possible bound is achieved. 
We further examine a special case in which each hospital has a utility that is proportional to the size of each contract. 
We establish two mechanisms that achieve better approximation ratios and, 
Notably, by sacrificing strategy-proofness, the ratio is bounded by a constant. 
Table~\ref{tab:summary} summarizes the results on approximation ratios.

\begin{table}[tb]
  \caption{Summary of the Results: UB and LB stand for upper and lower bounds, respectively.
    Let $\smax$ be the maximum fraction of a wage to budget among the given contracts and $|D|$ be the number of doctors. 
  }\label{tab:summary}
  \centering
  \subcaption{Non-Strategy-Proof Mechanisms: Thm.~\ref{thm:knapsack_lower} holds if $\smax>1/2$.}
  \label{table:results_nonsp}
  \scalebox{0.92}{%
    {\renewcommand\arraystretch{1.2}
      \setlength{\tabcolsep}{4pt}
      \begin{tabular}[htbp]{|l|l|}\hline
        \textbf{Hosp. Utils.}                                  & \textbf{Approximation Ratio}                     \\\hline
        Additive                           & UB: $\frac{1}{1-\smax}$ (Alg.\ \ref{alg:knapsack} and Thm.~\ref{thm:knapsack})\\
        (incl. Proportional)                                                         & LB: $\frac{1}{1-\smax}$ (Thm.~\ref{thm:knapsack_lower})\\\hline 
        \multirow{2}{*}{Proportional}& UB: $1.62$ (Alg.\ \ref{alg:prop_knapsack} and Thm.~\ref{thm:prop_knapsack})\\
                                                                                     & LB: $1.62$ (Thm.~\ref{thm:prop_nonexist})\\\hline
      \end{tabular}}}

  \smallskip
  
  \subcaption{Strategy-Proof Mechanisms.}
  \label{table:results_sp}
\scalebox{0.92}{%
{\renewcommand\arraystretch{1.2}
\setlength{\tabcolsep}{4pt}
  \begin{tabular}[htbp]{|l|l|l|}\hline
    \textbf{Hosp. Utils.}        & \textbf{Approximation Ratio}                     \\\hline
    Additive                     & UB: $\lceil\frac{1+\ln(|D|-1)}{1-\smax}\rceil$ (Alg.\ \ref{alg:knapsack_sp} and Thm.~\ref{thm:knapsack_sp})\\
    (incl. Proportional)                                                          & LB: open\\\hline
    \multirow{2}{*}{Proportional}& UB: $\frac{1}{1-\smax}$                        (Alg.\ \ref{alg:prop_knapsack_sp} and Thm.~\ref{thm:prop_knapsack_sp})\\
                                                                                   & LB: open\\\hline
  \end{tabular}}}
\end{table}



Let us finally note that there is a certain amount of recent studies
on two-sided matchings in the AI 
community, although this literature has been established mainly in the field across algorithms and economics. 
\citeauthor{drummond:ijcai:2013}~\shortcite{drummond:ijcai:2013,drummond:aaai:2014}
examine preference
elicitation procedures for two-sided matching.
In the context of \textit{mechanism design},
\citeauthor{hosseini:aaai:2015}~\shortcite{hosseini:aaai:2015}
consider a mechanism for a situation where
agents' preferences dynamically change.
\citeauthor{kurata:jair:2017}~\shortcite{kurata:jair:2017}
deal with strategy-proof mechanisms
for affirmative action in school choice programs (diversity constraints), while
\citeauthor{Goto:aij:2016}~\shortcite{Goto:aij:2016}
handle regional constraints, e.g.,
regional minimum/maximum quotas are imposed on hospitals in urban areas
so that more doctors are allocated to rural areas.

\section{Model}
This section describes a model for two-sided matchings with budget constraints. 
A market is a tuple $(D,H,X,\succ_D,u,s)$ where each component is defined as follows: 
There is a finite set of doctors $D=\{d_1,\ldots,d_n\}$ and a finite set of hospitals $H=\{h_1,\ldots,h_m\}$. 
Let $X$ denote a finite set of contracts where each contract $x\in X$ is associated with a doctor $x_D\in D$ and a hospital $x_H\in H$.
Let $s: X\to\mathbb{R}_{+}$ be the size function, where $\mathbb{R}_{+}$ is the set of nonnegative real numbers. 
A contract $x$ means that hospital $x_H$ offers $s(x)$ fraction of its budget as a wage to doctor $x_D$. 
For any subset of contracts $X'\subseteq X$, 
let $X'_d$ denote $\{x\in X' \mid x_D=d\}$ and $X'_h$ denote $\{x\in X' \mid x_H=h\}$. 
For notational simplicity, let $s(X')=\sum_{x\in X'}s(x)$ for $X'\subseteq X$.

Let $\succ_D={(\succ_d)}_{d\in D}$ denote the doctors' preference profile 
where $\succ_d$ is the strict relation of $d\in D$ over $X_d\cup\{\emptyset\}$; that is, 
$x\succ_d x'$ means that $d$ strictly prefers $x$ to $x'$. $\emptyset$ indicates a \textit{null} contract.
%
Let $u: X\to\mathbb{R}_+$ be the utility function.
We assume that the utility of each hospital is \emph{additive}; that is, 
for any two sets of contracts $X', X''\subseteq X_h$, hospital $h$ prefers $X'$ to $X''$ if and only if $\sum_{x\in X'}u(x)>\sum_{x\in X''}u(x)$ holds. 
In what follows, we denote $\sum_{x\in X_h'}u(x)$ by $u(X_h')$ for simplicity.

We call a subset of contracts $X'\subseteq X$ a \emph{matching} if
$|X'_d| \le 1$ for all $d\in D$ and $s(X'_h)\le 1$ for all $h\in H$.
Given a matching $X'$,
another matching $X''\subseteq X_h$ for a hospital $h$
is an \emph{$\alpha$-blocking coalition} if
\begin{enumerate}
\item $x\succ_{x_D} x'$ for any $x\in X''\setminus X'$ and $x'\in X_{x_D}'$, and
\item $u(X'')>\alpha\cdot u(X_h')$. 
\end{enumerate}
We then obtain a stability concept. 
\begin{definition}[$\alpha$-stability~\cite{arkin2009}]\label{def:stability}
  We say a matching $X'$ is \emph{$\alpha$-stable} if
  there exists no $\alpha$-blocking coalition. 
\end{definition}
As we can see, $1$-stability is equivalent to the standard stability concept.
Intuitively, the hospitals are willing to change the assignments only for a multiplicative improvement of $\alpha$. 
This idea regards the value of $\alpha$ as a switching cost for the hospitals.
Note that we can also define the second condition 
in an additive manner, i.e., $u(X'')>u(X_h')+\alpha$. 
However, since this does not satisfy scale invariance, the value of $\alpha$ is not well bounded.
In fact, when a market has no $\alpha$-stable matching in the additive sense, 
a market with the hospitals' utilities that are multiplied by 100 has no $100\alpha$-stable matching.

A mechanism is a function that takes a profile of 
doctors' preferences as input and returns a matching $X'$. 
We say a mechanism is $\alpha$-stable if
it always produces an $\alpha$-stable matching for certain $\alpha$. 
We also say a mechanism is \emph{strategy-proof} for doctors if
no doctor ever has any incentive to misreport her preference,
regardless of what the other doctors report.

To design and analyze a mechanism, 
we modify the matching with contract framework~\cite{Hatfield:AER:2005} for our approximate stability concept. 
It uses choice functions $\Ch_D: 2^X \rightarrow 2^X$ and $\Ch_H: 2^X \rightarrow 2^X$. 
For each doctor $d$, its choice $\Ch_d(X_d')$ is $\{x\}$ such that $x$ is the most preferred contract within $X'_d$. 
We assume $\Ch_d(X'_d)=\emptyset$ if $\emptyset\succ_d x$ for all $x\in X'_d$. 
Then, the choice function of all doctors $\Ch_D(X')$ is given as $\bigcup_{d \in D} \Ch_d(X_d')$.

Similarly, the choice of all hospitals $\Ch_H(X')$ is
$\bigcup_{h \in H} \Ch_h(X_h')$, where $\Ch_h(X_h')$ is a subset of $X'_h$. 
There are alternative ways to define the choice function of each hospital $\Ch_h$. 
As we discuss later, the mechanisms considered in this paper can be expressed
by the generalized DA with different formulations of $\Ch_H$.

In particular, we construct it from a class of \textit{sequential choice functions} that receives a sequence of contracts.
\begin{definition}
Let $X_h=\{x_1,\dots,x_n\}$ and $\mathcal{S}_n$ be the symmetric group of degree $n$.
A \emph{sequential choice function} on $X_h$ is a function $\Ch_h$ over a sequence of distinct contracts in $X_h$ such that
\begin{align}\label{eq:SC1}
\Ch_h(x_{\sigma(1)},\dots, x_{\sigma(k)}) & \subseteq\{x_{\sigma(1)},\dots, x_{\sigma(k)}\}\ \\ \notag
                                          & \text{for all } \sigma\in\mathcal{S}_n \text{ and } 0\le k\le n, 
\end{align}
and
\begin{align}\label{eq:SC2}
\begin{array}{l}
\{x_{\sigma(1)},\dots, x_{\sigma(k)}\}  \setminus\Ch_h(x_{\sigma(1)},\dots, x_{\sigma(k)}) \\ 
\subseteq \{x_{\sigma(1)},\dots, x_{\sigma(k+1)}\} \setminus\Ch_h(x_{\sigma(1)},\dots, x_{\sigma(k+1)})
\end{array}\\
\text{for all }\sigma\in\mathcal{S}_n \text{ and } 0\le k< n. \notag
\end{align}
\end{definition}
Here \eqref{eq:SC1} signifies that the choice must be taken from available contracts 
and \eqref{eq:SC2} signifies that, if a contract $x$ is rejected at some point, then it is also rejected ever after. 


This framework characterizes a class of mechanisms called the \emph{generalized DA mechanism}. 
If a mechanism---specifically, the choice function of every hospital---satisfies 
three properties, \textit{substitutability}, \textit{irrelevance of rejected contracts},
and \textit{law of aggregate demand}, then it always finds a ``stable'' allocation and is strategy-proof for doctors~\cite{Hatfield:AER:2005}.
This result is preserved because a sequential choice function can induce a (standard) choice function if it is \textit{order-invariant}. 
\begin{definition}[Order-invariance] 
  A sequential choice function $\Ch_h$ on $X_h=\{x_1,\dots,x_n\}$ is 
  \emph{order-invariant} if
  \begin{align*}
    \Ch_h(x_{\sigma(1)},\dots, x_{\sigma(k)})=\Ch_h(x_{\tau(1)},\dots, x_{\tau(k)})
  \end{align*}
  holds for any $\sigma,\tau\in\mathcal{S}_n$ and $0\le k\le n$
  such that $\{x_{\sigma(1)},\dots, x_{\sigma(k)}\}=\{x_{\tau(1)},\dots, x_{\tau(k)}\}$.
\end{definition}
Thus, 
if a sequential choice function is order-invariant,
the induced choice function satisfies substitutability.
The induced choice function satisfies law of aggregate demand if it further satisfies size-monotonicity.%
\footnote{If a choice function satisfies substitutability
  and law of aggregate demand simultaneously, it also satisfies irrelevance of rejected contracts~\cite{aizerman1981gto}.} 
\begin{definition}[Size-monotonicity]
  A sequential choice function $\Ch_h$ on $X_h=\{x_1,\dots,x_n\}$ is 
  \emph{size-monotone} if
  \begin{align*}
    |\Ch_h(x_{\sigma(1)},\dots, x_{\sigma(i)})|\le |\Ch_h(x_{\sigma(1)},\dots, x_{\sigma(j)})|
  \end{align*}
  holds for any $\sigma\in\mathcal{S}_n$ and $0\le i\le j\le n$.
\end{definition}
 
Here, strategy-proofness is still characterized by order-invariance and size-monotonicity
because hospital choice does not affect doctor preferences. 
However, even if a sequential choice function for each hospital is order-invariant, 
the mechanism employed with the function does not always find a stable matching 
because such a matching may not exist in the presence of budget constraints.

Finally, for Lemma~\ref{lem:matroid-approximation}, let us briefly explain a \emph{matroid}, 
which is a set system $(X,\mathcal{F})$ with the following properties: 
(i) $\emptyset\in\mathcal{F}$;
(ii) $F\subseteq G\in\mathcal{F}$ implies $F\in\mathcal{F}$; and 
(iii) $F,G\in\mathcal{F}$, $|F|<|G|$ implies the existence of $x\in G\setminus F$ such that $F\cup\{x\}\in\mathcal{F}$.
Notice that a subset \(F\) of \(X\) is called an \emph{independent set} if \(F\) belongs to \(\mathcal{F}\).

\section{Negative Results}
\label{sec:negative results}

The nonexistence of stable matchings raises the issue of the complexity of deciding the existence of a ($1$-)stable matching. 
\citeauthor{mcdermid2010}~\shortcite{mcdermid2010}
considered a special case of our model and proved NP-hardness. 
\citeauthor{HISY2017}~\shortcite{HISY2017}
examined a similar model to ours
and proved that the existence problem is $\Sigma_2^P$-complete. 
%
With approximately stable matchings in Definition~\ref{def:stability}, 
we provide the following hardness result. 
\begin{theorem}\label{thm:np-hard}
  For any constant $\alpha>1$, distinguishing whether a given market
  has a $1$-stable matching or
  has no $\alpha$-stable matching is NP-complete. 
\end{theorem}
We prove the NP-completeness by reducing from the \emph{subset sum problem}~\cite{garey1979cai}, 
though details are omitted due to space limitations. 
Furthermore, to obtain an approximately stable matching, 
the amount of approximation does not fall $\frac{1}{1-\smax}$, 
where $\smax$ indicates the largest wage size among the given contracts, i.e., $\smax=\max_{x\in X}s(x)$.

\begin{theorem}\label{thm:knapsack_lower}
For any positive reals $1/2<\smax<1$ and $\epsilon$, 
there exists a market such that
$s(x)\le\smax$ for all $x$ and
no $\left(\frac{1}{1-\smax}-\epsilon\right)$-stable matching exists.
\end{theorem}
We provide the instance below to prove this theorem due to the space limitation
(Please see the full version for additional details).
\begin{proof} 
Let $\delta$ be a sufficiently small positive real and $\alpha=1/(1-\smax)$.
Let $n=\lfloor 1/\delta^3\rfloor$, $m=\lceil 1/\smax\rceil$, and $l=\lfloor(1-\smax)/\delta\rfloor+1$.
 
Consider a market with $n+1$ doctors and $n-l+1$ hospitals: $D=\{d^*,d_1,\dots,d_{n}\}$ and $H=\{h^*,h_{l+1},\dots,h_n\}$. 
The set of contracts $X$ is given as
\begin{align*}
\{x^*,x^1,\dots,x^{n}\}\cup
\left\{y^{i,j}\mid \!\!\!\begin{array}{l}i=1,\dots,n,\\j=\max\{i,l+1\},\dots,n\end{array}\!\!\!\right\}, 
\end{align*}
where
\begin{align*}
\scalebox{0.83}{$\displaystyle
\arraycolsep=2pt
\begin{array}{lllll}
x^*_D=d^*,    & x^*_H=h^*,    & u(x)=\smax,               & s(x)=\smax,     &\\
x^i_D=d_i,    & x^i_H=h^*,    & u(x^i)=i\cdot\delta^3,    & s(x^i)=\delta   &(i=1,\dots,n),\ \mathrm{and}\\
y^{i,j}_D=d_i,& y^{i,j}_H=h_j,& u(y^{i,j})=\alpha^{-i},   & s(y^{i,j})=\smax&(\substack{i=1,\dots,n,\hfill\\j=\max\{i,l+1\},\dots,n}).
\end{array}$}
\end{align*}
The doctors' preferences are given as follows:
\begin{align*}
&\scalebox{0.85}{$\displaystyle \succ_{d_0}:~x^*,$}\\
&\scalebox{0.85}{$\displaystyle \succ_{d_i}:~x^i\succ_{d_i}y^{i,l+1}\succ_{d_i}\dots\succ_{d_i}y^{i,n} \quad(i=1,\dots,l),\ \mathrm{and}$}\\
&\scalebox{0.85}{$\displaystyle \succ_{d_i}:~y^{i,i}\succ_{d_i}x^i\succ_{d_i}y^{i,i+1}\succ_{d_i}\dots\succ_{d_i}y^{i,n} \quad(i=l+1,\dots,n).$}
\end{align*}
Then, we claim that there exists no $\left(\alpha-\epsilon\right)$-stable matching 
in the market instance. 
\end{proof}

Note that we have no non-trivial lower bound of the approximation ratio when $\smax\le 1/2$. 


\section{Framework for Approximate Stability}

This section introduces a new property, which we call \textit{$\alpha$-approximation}, 
and modifies the matching with contract framework~\cite{Hatfield:AER:2005}. 
Let $\mathcal{A}(X_h)$ be the set of sequences of distinct contracts in $X_h$.
A sequential choice function $\Ch_h$ is called $\alpha$-approximate 
if $\Ch_h$ always chooses a feasible subset of given contracts with utility at least $1/\alpha$ times the optimal one. 
\begin{definition}[$\alpha$-approximation]\label{def:approx}
  Given $\alpha\ge 1$, for a hospital $h\in H$,
  a sequential choice function $\Ch_h$ on $X_h$ is 
  \emph{$\alpha$-approximate} if
  $s(\Ch_h(\bm{x}))\le 1$ and 
  \begin{align}
  \scalebox{0.95}{$\displaystyle
   \alpha\cdot u(\Ch_h(\bm{x}))
   \ge \max\!\left\{u(\tilde{X})\!\mid\!\!
   \begin{array}{l}
      s(\tilde{X})\le 1,\\
      |\tilde{X}_d|\le 1~(\forall d\in D),\\
      \tilde{X}\subseteq\{x_1,\dots,x_k\}
  \end{array}\!\!\!\right\}$}\label{eq:approx}
  \end{align}
  holds for any $\bm{x}=(x_1,\dots, x_k)\in\mathcal{A}(X_h)$ such that $|\Ch_h(\bm{x})_d|\le 1~(\forall d\in D)$.
\end{definition}
We note that 
a sequential choice function can be viewed as an online algorithm for a \emph{removable knapsack problem} and
an $\alpha$-competitive online algorithm~\cite{komm2016} implies an $\alpha$-approximate sequential choice function.

Let us next introduce a modified version of the generalized DA mechanism,
which is formally given as Algorithm~\ref{alg:GDA}.
This mechanism uses choice functions for the doctors $(\Ch_d)_{d\in D}$ and 
sequential choice functions for the hospitals $(\Ch_h)_{h\in H}$. 
Initially, each doctor is set to be unmatched. 
Then, an unmatched doctor $d$ proposes to her most preferred contract $x$ that has not been rejected 
and the proposed hospital $x_H$ chooses a set of contracts $\Ch_{x_H}(\bm{a}^{x_H})$, 
where $\bm{a}^{x_H}$ is the sequence of proposed contracts for hospital $x_H$.
The proposal procedure continues as long as an unmatched doctor has a non-rejected acceptable contract.

\begin{algorithm}[htb]
  \SetKwInOut{Input}{input}\Input{$X,{(\Ch_d)}_{d\in D},{(\Ch_h)}_{h\in H}$\quad\textbf{output:} matching}
  \caption{Generalized DA}\label{alg:GDA}
  $X'\ot\emptyset$, $R\ot\emptyset$, $I\ot D$\;
  $\bm{a}^h\ot ()$ for all $h\in H$\;
  \While{$I\ne\emptyset$}{
    pick $d\in I$ arbitrarily\;\label{line:pick_contract}
    let $\{x\}=\Ch_d(X_d\setminus R_{d})$\;
    append $x$ to the end of $\bm{a}^{x_H}$\;
    $X'\ot \Ch_{x_H}(\bm{a}^{x_H})\cup \bigcup_{h\in H\setminus\{x_H\}}X'_h$\;
    $R \ot R\cup\{x\}$\;
    $I\ot\{d\in D\mid X_{d}'=\emptyset\ \text{and}\ \Ch_d(X_d\setminus R_{d})\ne\emptyset\}$\;
  }
  \Return $X'$\;
\end{algorithm}

%
We next prove that $\alpha$-approximate choice functions lead to approximately stable matchings. 
Note that although the output of the generalized DA depends on which doctor is selected in line~\ref{line:pick_contract}, 
Theorem~\ref{thm:GDA} holds regardless of the output of the mechanism. 
\begin{theorem}\label{thm:GDA}
  If the sequential choice function $\Ch_h$ is $\alpha$-approximate for every hospital $h\in H$,
  then the generalized DA produces an $\alpha$-stable matching $X'\subseteq X$. 
\end{theorem}

\begin{proof}
Let $\bm{a}^h=(x^h_1,\dots,x^h_{i_h})$ for each $h\in H$.
Note that $X'$ is a matching and
$X'_h=\Ch_h(\bm{a}^{h})=\Ch_h(x^{h}_1,\dots,x^h_{n})$ by the definition of the algorithm.

We prove that $X'$ is $\alpha$-stable by contradiction.
Suppose that $X''\subseteq X_h$ is an $\alpha$-blocking coalition for $X'$.
Then, $X''\subseteq\{x^h_1,\dots,x^h_{i_h}\}$ since $x\succ_{x_D} x'$ holds for any $x\in X''\setminus X'$ and $x'\in X_{x_D}'$.
Thus, we have 
\begin{align*}
u(X'')&>\alpha\cdot u(X_h')=\alpha\cdot u(\Ch_h(\bm{a}^h))\\
&\ge \max\left\{u(\tilde{X})\mid \!
\scalebox{0.95}{$\displaystyle
\begin{array}{l}
s(\tilde{X})\le 1,\\
|\tilde{X}_d|\le 1~(\forall d\in D),\\
\tilde{X}\subseteq\{x^h_{1},\dots,x^h_{i_h}\}
\end{array}$}
\!\!\right\}\ge u(X''),
\end{align*}
which is a contradiction.
Here, the first inequality holds since $X''$ is an $\alpha$-blocking coalition
and the second inequality holds since $\Ch_h$ is $\alpha$-approximate and $|\Ch_h(\bm{a}^h)_d|\le 1~(\forall d\in D)$.
\end{proof}

It must be emphasized that this framework can be applied to other types of constraints,
e.g., \textit{matroid intersection} constraints. 
The typical example is diversity constraints in school choice programs 
where a school is required to balance the composition of students, 
typically in terms of socioeconomic status~\cite{kurata:jair:2017}. 
The framework also works even when each hospital has a \textit{submodular} utility. 
Designing such mechanisms is our immediate future work and 
it will be achieved with online algorithms 
for the corresponding problems~\cite{ashwinkumar2011,HKMG2014,chakrabarti2015,buchbinder2015,HMMC2017,CJTW2017}. 

\section{Approximately Stable Mechanisms}
In matching with constraints~\cite{kamada:aer:2015,Goto:aij:2016,kurata:jair:2017}, 
designing a desirable mechanism essentially tailors the choice functions for hospitals
to satisfy the necessary properties and constraints simultaneously. 
The design task is difficult because the given sequences of contracts for hospitals
depend on the (sequential) choice functions themselves and, hence, they are unpredictable.
We tackle this challenging task as an analogue to online algorithms for knapsack problems. 

\subsection{Non-Strategy-Proof Stable Mechanism}
This subsection proposes an approximately stable mechanism
that achieves the best possible stability bound, 
but is not strategy-proof for doctors. 
%
The sequential choice functions greedily choose contracts 
according to decreasing order of utility per size (i.e., utility per wage).
Formally, 
let us consider the functions defined in Algorithm~\ref{alg:knapsack}.\footnotemark
\footnotetext{When ties occur in the argmin, we break these ties by choosing the one with the smallest index.}
Note that we can compute the choice in $O(k\log k)$ time. 

\begin{algorithm}
\SetKwInOut{Input}{input}\Input{\small $(x_1,\dots,x_k)\in\mathcal{A}(X_h)$\quad\textbf{output:} $\Ch_h(x_1,\dots,x_k)$}
\caption{}\label{alg:knapsack}
\lIf{$k=0$}{\Return $\emptyset$}
let $Y\ot \Ch_h(x_1,\dots,x_{k-1})\cup\{x_k\}$\;
\While{$s(Y)>1$}{
  $Y\ot Y\setminus\{a\}$ where $a\in\argmin\big\{\frac{u(x)}{s(x)}\mid x\in Y\big\}$\;
}
\Return $Y$\;
\end{algorithm}

Let us next illustrate this mechanism via an example.
\begin{example} 
Consider a market with four doctors $D=\{d_1,d_2,d_3,d_4\}$ and two hospitals $H=\{h_1,h_2\}$.
The set of offered contracts $X$ is $\{x^{i,j}\mid i=1,\dots,4,\ j=1,2\}$, 
where $x^{i,j}_D=d_i$, $x^{i,j}_H=h_j$ ($i=1,\dots,4$, $j=1,2$).
We assume that $x^{i,1}\succ_{d_i} x^{i,2}$ for $i=1,2,3$ and $x^{4,2}\succ_{d_4} x^{4,1}$.
The contracts' utilities and sizes are provided in Table~\ref{tab:ex:1}. 

\begin{table}[htb]
\centering
\caption{Utilities and Sizes of the Contracts.}
\label{tab:ex:1}
\resizebox{0.9\columnwidth}{!}{%
\begin{tabular}[htbp]{c||c|c|c||c|c|c}
                    &\multicolumn{3}{|c||}{$x^{i,1}\in X_{h_1}$}&\multicolumn{3}{|c}{$x^{i,2}\in X_{h_2}$}\\\hline
                    &$u$   &$s$   &$u/s$   &$u$  &$s$&$u/s$\\\hline
$x^{1,j}\in X_{d_1}$&$111$ &$0.57$&$194.7$ &$30$ &$0.56$& $53.6$ \\
$x^{2,j}\in X_{d_2}$&$98$  &$0.50$&$196.0$ &$40$ &$0.55$& $72.7$ \\
$x^{3,j}\in X_{d_3}$&$83$  &$0.42$&$197.6$ &$10$ &$0.60$& $16.7$ \\ 
$x^{4,j}\in X_{d_4}$&$110$ &$0.55$&$200.0$ &$20$ &$0.45$& $44.4$
\end{tabular}}
\end{table}

Initially, we set $X'=R=\emptyset$, $I=\{d_1,\dots,d_4\}$, and $\bm{a}^{h_i}=()$ $(i=1,2)$.

Suppose that we select $d_1\in I$ at the beginning of the first iteration. 
Then, $d_1$ chooses her most preferred contract $x^{1,1}$ 
and we have $\bm{a}^{h_1}=(x^{1,1})$, $\bm{a}^{h_2}=()$, $X'=\{x^{1,1}\}$, $R=\{x^{1,1}\}$, $I=\{d_2,d_3,d_4\}$ 
at the end of the first iteration. 

Suppose that we pick $d_2\in I$ at the beginning of the second iteration. 
Then, $d_2$ chooses her most preferred contract $x^{2,1}$.
As $h_1$ cannot keep $x^{1,1}$ and $x^{2,1}$ together, it accepts $x^{2,1}$ and rejects $x^{1,1}$.
Thus, we have $\bm{a}^{h_1}=(x^{1,1},x^{2,1})$, $\bm{a}^{h_2}=()$, $X'=\{x^{2,1}\}$, $R=\{x^{1,1},x^{2,1}\}$, $I=\{d_1,d_3,d_4\}$
at the end of the second iteration. 

Similarly, suppose that we choose a doctor with the smallest index in $I$ at the beginning of each iteration.
Our algorithm then works as in Table~\ref{tab:ex:1:round}. 

\begin{table}[hbt]
\centering
\caption{Our Algorithm's Procedure.}
\label{tab:ex:1:round}
\setlength{\tabcolsep}{1pt}
\resizebox{\linewidth}{!}{%
\begin{tabular}[htbp]{cccccc}
iter. & $X'$                          & $I$                   & $\bm{a}^{h_1}$              & $\bm{a}^{h_2}$ \\\hline
0     & $\emptyset$                   & $\{d_1,d_2,d_3,d_4\}$ & $()$                        & $()$      \\ 
1     & $\{x^{1,1}\}$                 & $\{d_2,d_3,d_4\}$     & $(x^{1,1})$                 & $()$      \\
2     & $\{x^{2,1}\}$                 & $\{d_1,d_3,d_4\}$     & $(x^{1,1},x^{2,1})$         & $()$      \\
3     & $\{x^{1,2},x^{2,1}\}$         & $\{d_3,d_4\}$         & $(x^{1,1},x^{2,1})$         & $(x^{1,2})$  \\
4     & $\{x^{1,2},x^{2,1},x^{3,1}\}$ & $\{d_4\}$             & $(x^{1,1},x^{2,1},x^{3,1})$ & $(x^{1,2})$  \\
5     & $\{x^{1,2},x^{2,1},x^{3,1}\}$ & $\{d_4\}$             & $(x^{1,1},x^{2,1},x^{3,1})$ & $(x^{1,2},x^{4,2})$  \\
6     & $\{x^{1,2},x^{3,1},x^{4,1}\}$ & $\{d_2\}$             & $(x^{1,1},x^{2,1},x^{3,1},x^{4,1})$ & $(x^{1,2},x^{4,2})$  \\
7     & $\{x^{2,2},x^{3,1},x^{4,1}\}$ & $\emptyset$           & $(x^{1,1},x^{2,1},x^{3,1},x^{4,1})$ & $(x^{1,2},x^{4,2},x^{2,2})$ 
\end{tabular}}
\end{table}

To conclude, our algorithm returns the matching $X'=\{x^{2,2},x^{3,1},x^{4,1}\}$.
Although $X''=\{x^{1,1},x^{3,1}\}$ is a blocking coalition,
$X'$ is almost stable.
\end{example}

We argue that the sequential choice function $\Ch_h$ satisfies the following property: 
\begin{lemma}\label{lem:knapsack}
  The sequential choice function $\Ch_h$ defined in Algorithm~\ref{alg:knapsack} is $\frac{1}{1-\smax}$-approximate. 
\end{lemma}
\begin{proof} 
Let us consider a sequence of distinct contracts $(x_1,\dots,x_k)\in\mathcal{A}(X_h)$. 
Note that $\Ch_h$ satisfies the conditions \eqref{eq:SC1} and \eqref{eq:SC2}
since $\Ch_h(x_1,\dots,x_{k})\subseteq \Ch_h(x_1,\dots,x_{k-1})\cup\{x_k\}$.
If $\sum_{i=1}^k s(x_i)\le 1$, then the sequential choice function returns all the contracts 
and therefore satisfies~\eqref{eq:approx}. Thus, we assume that $\sum_{i=1}^k s(x_i)> 1$.

Let $\sigma\in\mathcal{S}_k$ be a permutation such that
\[
  u(x_{\sigma(1)})/s(x_{\sigma(1)})\ge\dots\ge u(x_{\sigma(k)})/s(x_{\sigma(k)}). 
\]
Let $\ell$ be the largest index such that $\sum_{i=1}^{\ell}s(x_{\sigma(i)})\le 1$.
Note that $\sum_{i=1}^{\ell+1}s(x_{\sigma(i)})> 1$.
Then, by a simple induction, we can see that $\{x_{\sigma(1)},\dots,x_{\sigma(\ell)}\}\subseteq \Ch_h(x_1,\dots,x_k)$.
Therefore, we obtain $u(\Ch_h(x_1,\dots,x_k))$ is at least
\begin{align*} 
&\sum_{i=1}^{\ell}u(x_{\sigma(i)}) \ge \frac{\sum_{i=1}^{\ell}s(x_{\sigma(i)})}{\sum_{i=1}^{\ell+1}s(x_{\sigma(i)})}\left(\sum_{i=1}^{\ell+1}u(x_{\sigma(i)})\right) \\
&= 
\left(1-\frac{s(x_{\sigma(\ell+1)})}{\sum_{i=1}^{\ell+1}s(x_{\sigma(i)})}\right)
\left(\sum_{i=1}^{\ell+1}u(x_{\sigma(i)})\right)\\
&\ge (1-\smax)\left(\sum_{i=1}^{\ell+1}u(x_{\sigma(i)})\right)\\
&\ge (1-\smax)\cdot\max\left\{u(\tilde{X})\mid\! \begin{array}{l}s(\tilde{X})\le 1,\\\tilde{X}\subseteq\{x_1,\dots,x_k\}\end{array}\!\!\right\}.
\end{align*}
Here, the first inequality holds since $\frac{\sum_{i=1}^{j}u(x_{\sigma(i)})}{\sum_{i=1}^{j}s(x_{\sigma(i)})}$ is monotone nonincreasing for $j$
and the second inequality holds by $\sum_{i=1}^{\ell+1}s(x_{\sigma(i)})>1$ and $s(x_{\sigma(\ell+1)})\le\smax$.
Thus, it is $\frac{1}{1-\smax}$-approximate. 
\end{proof}

Now, we obtain the following theorem:
\begin{theorem}\label{thm:knapsack}
For any market, the generalized DA mechanism with the choice functions defined in Algorithm~\ref{alg:knapsack} 
produces a $\frac{1}{1-\smax}$-stable matching.
In addition, the mechanism can be implemented to run in $O(|X|\log|X|)$ time. 
\end{theorem}

This mechanism is not strategy-proof for doctors 
(See Proposition~1 
in the full version).

\subsection{Strategy-Proof Stable Mechanism}
This subsection proposes another stable mechanism 
that does not achieve the best possible stability bound, but is strategy-proof for doctors. 
%
Let us consider a set system $(X_h,\mathcal{F}_h)$, where 
\[
\mathcal{F}_h=\left\{X'\subseteq X_h\mid \begin{array}{r}\left|\left\{x\in X'\mid s(x)>\frac{1}{t\gamma}\right\}\right|\le t\\ (t=1,\dots,|D|)\end{array}\right\}
\]
and $\gamma=\left\lceil\frac{1+\ln(|D|-1)}{1-\smax}\right\rceil$. 
We can easily check that the set system $(X_h,\mathcal{F}_h)$ is a transversal matroid for each $h\in H$~\cite{oxley1992mt}.
%
In addition, we have $s(X')\le 1$ if $X'\in\mathcal{F}_h$ and $|X'_d|\le 1~(\forall d\in D)$ because
\begin{align*}
  s(X')
  &\le \smax+\sum_{t=1}^{|D|-1}\frac{1}{t\gamma} \le \smax+\frac{1}{\gamma}\cdot\left(1+\int_1^{|D|-1}\frac{dx}{x}\right)\\
  &\le \smax+\frac{1+\ln(|D|-1)}{\gamma}\le 1.
\end{align*}

The second sequential choice function $\Ch_h$ chooses the maximum utility independent set by a greedy algorithm.
The formal definition is given in Algorithm~\ref{alg:knapsack_sp}. 
\begin{algorithm}
  \SetKwInOut{Input}{input}\Input{\small $(x_1,\dots,x_k)\in\mathcal{A}(X_h)$\quad\textbf{output:} $\Ch_h(x_1,\dots,x_k)$}
  \caption{}\label{alg:knapsack_sp}
  \lIf{$k=0$}{\Return $\emptyset$}
  let $Y\ot \Ch_h(x_1,\dots,x_{k-1})\cup\{x_k\}$\;
  \If{$Y\not\in\mathcal{F}_h$}{
    let $a\in\argmin\{u(x)\mid x\in Y,~Y\setminus\{x\}\in\mathcal{F}_h\}$\;
    $Y\ot Y\setminus\{a\}$\;
  }
  \Return $Y$\;
\end{algorithm}

We claim that the choice function satisfies the following property. 
\begin{lemma}\label{lem:matroid-approximation}
  For each hospital $h$, the sequential choice function $\Ch_h$ defined in Algorithm~\ref{alg:knapsack_sp} is 
  $\gamma~(=\left\lceil\frac{1+\ln(|D|-1)}{1-\smax}\right\rceil)$-approximate. 
\end{lemma}
\begin{proof}
  Let $(x_1,\dots,x_k)$ be the input of the choice function. 
  We prove that, for any $X^*\subseteq\{x_1,\dots,x_k\}$,
  \[\gamma\cdot u(\Ch_h(x_1,\dots,x_k))\ge u(X^*)\]
  holds when $s(X^*)\le 1$ and $|X^*_d|\le 1~(\forall d\in D)$.
  Let us consider a set of contracts $X^*=\{y_1,\dots,y_\ell\}$ where $s(y_1)\ge\dots\ge s(y_\ell)$.
  Then, we have $s(y_i)\le 1/i$ for $i=1,\dots,\ell$ because $1\ge s(y_1)+\dots+s(y_i)\ge i\cdot s(y_i)$.
  Thus, we can partition $X^*$ into $\gamma$ sets $X^1,\dots,X^{\gamma}$ that satisfy $X^i\in\mathcal{F}_h$ for all $i=1,\dots,\gamma$.
  Indeed, we can obtain the partition by setting $X^i=\{y_i,y_{\gamma+i},y_{2\gamma+i},\dots,y_{\lfloor(\ell-i)/\gamma\rfloor\gamma+i}\}$.
  As $u(\Ch_h(x_1,\dots,x_k))\ge u(X^i)$, we obtain
  \begin{align*}
    \gamma\cdot u(\Ch_h(x_1,\dots,x_k))\ge \sum_{i=1}^\gamma u(X^i)=u(X^*),
  \end{align*}
  which proves the theorem.
\end{proof}

\begin{lemma}\label{lem:matroid-sp}
  For each hospital $h$, the sequential choice function $\Ch_h$ defined in Algorithm~\ref{alg:knapsack_sp} 
  satisfies order-invariance and size-monotonicity. 
\end{lemma}

\begin{proof}
The choice function is clearly order-invariant, 
since the greedy algorithm for matroids, which finds the maximum utility independent set,
is not affected by the input sequence order~\cite{babaioff2009}. 
Moreover, as $|\Ch_h(x_1,\dots,x_{k})|$ is the rank of $\{x_1,\dots,x_k\}$ in the matroid $\mathcal{F}_h$,
the choice function is size-monotone.  
\end{proof}

Now, we obtain the following theorem: 
\begin{theorem}\label{thm:knapsack_sp}
For any market, the generalized DA mechanism with the choice functions defined in Algorithm~\ref{alg:knapsack_sp} 
is strategy-proof for doctors and produces a $\left\lceil\frac{1+\ln(|D|-1)}{1-\smax}\right\rceil$-stable matching.
In addition, the mechanism can be implemented to run in $O(|X|^2)$ time. 
\end{theorem}

\section{Proportional Utilities}
This section examines a special case 
in which each hospital has a utility over a set of contracts 
that is proportional to its size. 
Formally, for every $h\in H$ and $x\in X_h$, 
$u(x)=\beta_h\cdot s(x)$ holds where $\beta_h~(>0)$ is a constant that only depends on the hospital. 
We call such a market a \emph{proportional market}. 
W.l.o.g., 
we may assume that $\beta_h=1$ for all $h\in H$.

To admit an approximately stable matching, 
the amount of approximation improves to a constant. 
We first provide an intractability result.
\begin{theorem}\label{thm:prop_nonexist}
  For any $\epsilon>0$,
  there exists a proportional market with no $\left(\frac{1+\sqrt{5}}{2}-\epsilon\right)$-stable matching.
\end{theorem}

\begin{proof}
Let $\phi=\frac{1+\sqrt{5}}{2}$. 
Consider a market with three doctors $D=\{d_1,d_2,d_3\}$ and two hospitals $H=\{h_1,h_2\}$.
The set of contracts is $X=\{x,y,\hat{y},z,\hat{z}\}$, where 
\begin{align*}
  \begin{array}{llll}
    x_D=d_1,      & x_H=h_1,      & s(x)=1/\phi+\epsilon/2,\\
    y_D=d_2,      & y_H=h_1,      & s(y)=(1/\phi)^2,\\
    \hat{y}_D=d_2,& \hat{y}_H=h_2,& s(\hat{y})=1,\\
    z_D=d_3,      & z_H=h_1,      & s(z)=1/\phi,\ \mathrm{and}\\
    \hat{z}_D=d_3,& \hat{z}_H=h_2,& s(\hat{z})=\epsilon.
  \end{array}
\end{align*}
The doctors' preferences are given as follows: 
\begin{align*}
\succ_{d_1}:~x,\quad \succ_{d_2}:~y\succ_{d_2}\hat{y},\quad \succ_{d_3}:~\hat{z}\succ_{d_3}z.
\end{align*}

If $X'=\{x,\hat{y}\}$, then $X''=\{y,z\}$ is a $(\phi-\epsilon)$-blocking coalition for $h_1$
since $\frac{1/\phi+(1/\phi)^2}{1/\phi+\epsilon/2}=\frac{\phi}{1+\epsilon\phi/2}<\phi-\epsilon$.
Similarly, we can easily check that
there exists a $(\phi-\epsilon)$-blocking coalition for any matching
and hence this market has no $(\phi-\epsilon)$-stable matching.
\end{proof}
However, distinguishing whether a given proportional market
has $\left(\frac{1+\sqrt{5}}{2}-\epsilon\right)$-stable matching or not is NP-hard
(See Theorem~9 
in the full version). 

\subsection{Non-Strategy-Proof $\mathbf{1.62}$-Stable Mechanism}
This subsection proposes a $\frac{1+\sqrt{5}}{2}~(\approx 1.62)$-stable mechanism. 
We denote the golden ratio by $\phi~(=\frac{1+\sqrt{5}}{2})$. 
Note that $\phi+1=\phi^2$ holds. 
To obtain the mechanism, we apply an online algorithm provided by \citeauthor{iwama2002rok}~\shortcite{iwama2002rok}
for the (proportional case) removable knapsack problem as a sequential choice function. 
Let us divide the set of contracts into three classes---\emph{small}, \emph{medium}, and \emph{large}. 
These are defined as 
\begin{align*}
   S&=\{x\in X\mid s(x)\le 1-1/\phi\},\\
   M&=\{x\in X\mid 1-1/\phi<s(x)<1/\phi\},\ \mathrm{and} \\
   L&=\{x\in X\mid s(x)\ge 1/\phi\}. 
\end{align*}
%
%
It should be noted that, if the algorithm chooses a set of contracts with a total size of at least $1/\phi$, 
then it is $\phi$-approximate since
\begin{align*}
  \phi \cdot  u&(\Ch_h(\bm{x}))\ge 1 
               \ge\max\left\{u(\tilde{X})\mid\!\!
               \begin{array}{l}
               s(\tilde{X})\le 1,\\
               \tilde{X}\subseteq \{x_1,\dots,x_k\}
               \end{array}\!\!\!\right\}
\end{align*}
for any $\bm{x}=(x_1,\dots,x_k)\in\mathcal{A}(X_h)$. 

\begin{algorithm}
\SetKwInOut{Input}{input}\Input{\small $(x_1,\dots,x_k)\in\mathcal{A}(X_h)$\quad\textbf{output:} $\Ch_h(x_1,\dots,x_k)$}
\caption{}\label{alg:prop_knapsack}
Initialize $S$, $M$, $L$\;
\lIf{$k=0$}{\Return $\emptyset$}
let $Y\ot \Ch_h(x_1,\dots,x_{k-1})$\;
\lIf{$s(Y)\ge 1/\phi$}{\Return $Y$}
\lElseIf{$x_k\in L$}{\Return $\{x_k\}$}
\Else{
    $Y\ot Y\cup\{x_k\}$\;
    \If{$s(Y\cap M)>1$}{
      $Y\ot Y\setminus\{a\}$ where $a\in\argmax_{x\in Y\cap M}s(x)$\;
    }
    \While{$s(Y)>1$}{
      $Y\ot Y\setminus\{a\}$ where $a\in\argmin_{x\in Y\cap S}s(x)$\;
    }
    \Return $Y$\;
  }
\end{algorithm}

Thus, if some large contracts have been given, then it is sufficient to select one of them. 
Also, if it can choose two medium contracts together, then it is sufficient to select both because $2(1-1/\phi)>1/\phi$.
If a given sequence of contracts contains medium contracts, but 
it cannot choose two medium contracts together and no large contract exists, 
the algorithm keeps the smallest medium contract.
The formal definition of the algorithm is given in Algorithm~\ref{alg:prop_knapsack}. 
Note that the function $\Ch_h$ 
satisfies \eqref{eq:SC1} and \eqref{eq:SC2} since $\Ch_h(x_1,\dots,x_{k})\subseteq \Ch_h(x_1,\dots,x_{k-1})\cup\{x_k\}$.

We obtain the next lemma through the techniques of \citeauthor{iwama2002rok}~\shortcite{iwama2002rok}. 
\begin{lemma}\label{lem:prop_knapsack}
  The sequential choice function defined in Algorithm~\ref{alg:prop_knapsack} is $\phi$-approximate. 
\end{lemma}
Now, with Theorem~\ref{thm:GDA}, we obtain the following theorem: 
\begin{theorem}\label{thm:prop_knapsack}
  For any proportional market, the generalized DA mechanism with the choice functions
  defined in Algorithm~\ref{alg:prop_knapsack} produces a $\phi$-stable matching.
  In addition, the mechanism can be implemented to run in $O(|X|\log|X|)$ time.
\end{theorem}

\subsection{Strategy-Proof Stable Mechanism Achieving the Best Possible Bound for Additive Utilities}

This subsection provides another stable mechanism for the proportional markets 
that does not have a constant stability bound, but is strategy-proof for doctors. 
We employ a sequential choice function $\Ch_h$ that greedily chooses the contracts
according to increasing order of size while keeping the budget constraint. 
Formally, we define this as Algorithm~\ref{alg:prop_knapsack_sp}. 

\begin{algorithm}
\SetKwInOut{Input}{input}\Input{\small $(x_1,\dots,x_k)\in\mathcal{A}(X_h)$\quad\textbf{output:} $\Ch_h(x_1,\dots,x_k)$}
\caption{}\label{alg:prop_knapsack_sp}
Initialize $Y\ot\emptyset$\;
Sort $\{x_1,\dots,x_k\}$ according to increasing order of size\;
\For{$i=1,2,\dots,k$}{
  let $x$ be the $i$th smallest contract\;
  \lIf{$s(Y)+s(x) \le 1$}{$Y\ot Y\cup\{x\}$}
}
\Return $Y$\;
\end{algorithm}

Let us first claim that the choice function is sequential. 
Since it clearly satisfies \eqref{eq:SC1}, let us check \eqref{eq:SC2}. 
W.l.o.g., 
we may assume that $s(x_1)<s(x_2)<\dots< s(x_{k})$ 
because the hospitals' choices depend not on the order of the input sequence, but on size.
If $\sum_{i=1}^{k}s(x_i)\le 1$, 
then $\{x_{1},\dots, x_{k}\}\setminus\Ch_h(x_{1},\dots, x_{k})$ must be empty and hence \eqref{eq:SC2} holds.
Thus, we assume that $\sum_{i=1}^{k}s(x_i)> 1$.

Let $\ell$ be the largest index such that $\sum_{i=1}^{\ell}s(x_i)\le 1$.
If the $(\ell+1)$st contract is chosen, the budget is exceeded. 
We then have $\Ch_h(x_{1},\dots, x_{k})=\{x_{1},\dots,x_{\ell}\}$ and 
$\{x_{1},\dots, x_{k}\}\setminus\Ch_h(x_{1},\dots, x_{k})=\{x_{\ell+1},\dots,x_{k}\}$.
With the $(k+1)$st contract $x_{k+1}$, if $s(x_{k+1})>s(x_{\ell+1})$, \eqref{eq:SC2} holds because $\Ch_h(x_{1},\dots, x_{k+1})=\{x_{1},\dots,x_{\ell}\}$. 
Otherwise, 
we have $\Ch_h(x_{1},\dots, x_{k+1})\subseteq \{x_1,\dots,x_\ell,x_{k+1}\}$ and hence \eqref{eq:SC2} holds.
Thus, the choice function $\Ch_h$ is sequential. 
We further claim that it satisfies the following properties: 
\begin{lemma}
  For each hospital $h$, the choice function defined in Algorithm~\ref{alg:prop_knapsack_sp} 
  is order-invariant, size-monotone, and $\frac{1}{1-\smax}$-approximate. 
\end{lemma}

\begin{proof}
  First, each choice function is clearly order-invariant from the definition.
  We thus assume without loss of generality that $s(x_1)<s(x_2)<\dots< s(x_{k})$ throughout this proof.

  Second, we observe that each choice function $\Ch_h$ is size-monotone. 
  To see this, it is sufficient to prove
  \begin{align}\label{eq:lad}
  |\Ch_h(x_1,\dots,x_k)|\le |\Ch_h(x_1,\dots,x_{k+1})|
  \end{align}
  for any sequence of contracts $(x_1,\dots,x_k,x_{k+1})$. 
  Recall that $\ell$ is the largest index such that the budget is not exceeded. 
  Then, we have $\Ch_h(x_{1},\dots, x_{k})=\{x_{1},\dots,x_{\ell}\}$.
  We consider two cases for the $(k+1)$st contract $x_{k+1}$. 
  If $s(x_{k+1})>s(x_{\ell})$, then \eqref{eq:lad} holds because $\{x_{1},\dots,x_{\ell}\}\subseteq \Ch_h(x_{1},\dots, x_{k+1})$. 
  Otherwise, 
  \eqref{eq:lad} also holds because $\{x_{1},\dots,x_{\ell-1},x_{k+1}\}\subseteq \Ch_h(x_{1},\dots, x_{k+1})$. 
  Accordingly, $\Ch_h$ is size-monotone. 

  Finally, we show that $\Ch_h$ is $\frac{1}{1-\smax}$-approximate.
  Let us assume that the input sequence is $\bm{x}=(x_1,\dots,x_k)$. 
  If $\sum_{i=1}^{k}s(x_i)\le 1$, we have $\Ch_h(\bm{x})=\{x_1,\dots,x_k\}$ and \\
    \begin{align*}
    \scalebox{0.85}{$\displaystyle
        u(\Ch_h(\bm{x}))
        =\sum_{i=1}^ks(x_i)
        =\max\left\{u(\tilde{X})\mid \!\!\!\begin{array}{l}s(\tilde{X})\le 1,\\\tilde{X}\subseteq \{x_1,\dots,x_k\}\end{array}\!\!\!\right\}.
        $}
    \end{align*}
  Otherwise, for the index $\ell$, we have
  \vspace{-1mm}
  {\small\[
    s(\Ch_h(\bm{x}))=\sum_{i=1}^{\ell}s(x_i)>1-s(x_{\ell+1})\ge 1-\smax
  \]}
  and hence
  \begin{align*}
    \scalebox{0.83}{$\displaystyle
    u(\Ch_h(\bm{x}))
    =(1-\smax)
     \ge (1-\smax)\cdot \max\left\{u(\tilde{X})\mid \!\!\!\begin{array}{l}s(\tilde{X})\le 1,\\\tilde{X}\subseteq \{x_1,\dots,x_k\}\end{array}\!\!\!\right\}.
    $}
  \end{align*}
  Thus, $\Ch_h$ is $\frac{1}{1-\smax}$-approximate. 
\end{proof}

Now, we obtain the following theorem:  
\begin{theorem}\label{thm:prop_knapsack_sp}
  For any proportional market,
  the generalized DA mechanism with the choice functions defined in Algorithm~\ref{alg:prop_knapsack_sp} 
  is strategy-proof for doctors and 
  produces a $\frac{1}{1-\smax}$-stable matching. 
  In addition, the mechanism can be implemented to run in $O(|X|\log|X|)$ time.
\end{theorem}

\section{Conclusion}

This paper examined matching with budget constraints and introduced a concept of approximately stable matchings. 
First, we extended the existing 
framework 
and proposed two novel mechanisms that return a stable matching in polynomial time: 
one is strategy-proof for doctors and the other is not. 
Second, we derived the bounds of the utilities' increment by coalitional deviation. 
The best possible bound is obtained by sacrificing strategy-proofness for additive utilities.
Finally, we examined a case in which hospitals have proportional utilities and 
found another two mechanisms that achieve better approximation ratios. 


\clearpage
\appendix
\section*{Omitted proofs}
In this section, we provide omitted proofs.

\newtheorem*{thm:np-hard}{Theorem~\ref{thm:np-hard}}
\begin{thm:np-hard}
  For any constant $\alpha>1$, distinguishing whether a given market
  has a $1$-stable matching or
  has no $\alpha$-stable matching is NP-complete. 
\end{thm:np-hard}
\begin{proof}[Proof of Theorem~\ref{thm:np-hard}]
We first observe that the problem is in the class NP.
A witness would be a $1$-stable matching $X'$.
Here $X'$ is a $1$-stable matching if\\
\scalebox{0.9}{\parbox{\linewidth}{
\begin{align*}
  u(X'_h)\ge \max\!\left\{u(\tilde{X})\!\mid\!\!
  {\small
  \begin{array}{l}
    s(\tilde{X})\le 1,\\
    |X_d'|\le 1~(\forall d\in D),\\
    \tilde{X}\subseteq\left\{x\in X_h\!\mid\! x\succeq_{x_D}x'\in X_{x_D}'\right\}
  \end{array}}\!\!\!\!
  \right\}.
\end{align*}}}
Since we only want to distinguish the existence of $1$-stable matching and nonexistence of $\alpha$-stable matching,
it is sufficient to check\\
\scalebox{0.9}{\parbox{\linewidth}{
\begin{align*}
  u(X'_h)\ge \frac{1}{\alpha} \max\!\left\{u(\tilde{X})\!\mid\!\!
  {\small
  \begin{array}{l}
    s(\tilde{X})\le 1,\\
    |X_d'|\le 1~(\forall d\in D),\\
    \tilde{X}\subseteq\left\{x\in X_h\!\mid\! x\succeq_{x_D}x'\in X_{x_D}'\right\}
  \end{array}}\!\!\!\!
  \right\}.
\end{align*}}}
We can check the condition for each $h$ in polynomial time 
because we can find  $\alpha$-approximate solution of \(\max\{u(\tilde{X})\mid s(\tilde{X})\le 1,~|\tilde{X}_d|\le 1~(\forall d\in D),~\tilde{X}\subseteq\{x\in X_h\mid x\succeq_{x_D}x'\in X_{x_D}'\}\}\) in polynomial time for any fixed $\alpha>1$~\cite{kellerer2004kp}.
Thus, the problem is in the class NP.

We then prove the NP-hardness by reducing the subset sum problem---given $n+1$ positive integers $a_1,\dots,a_n,t$,
ask whether there exists a subset \(I\subseteq \{1,\dots,n\}\) such that $\sum_{i\in I}a_i=t$. 

Let $a_1,\dots,a_n$ and $t$ be an instance of subset sum problem and let $a=\sum_{j=1}^n a_j$.
Also, let $m=\lceil 3\alpha^2\rceil$.
Without loss of generality, we assume that $a>t>0$.
Consider the following market:
{\small
\begin{itemize}
\item $D=\{d_i^j\mid i=1,\dots,n,~j=0,\dots,m\}$,
\item $H=\{h^+,h^-\}\cup\{h_i^j\mid i=1,\dots,n,~j=0,\dots,m-1\}$,
\item $X=\bigcup_{i=1}^nX_i$ where $X_i=\{r^{i,1},r^{i,2},r^{i,3}\}\cup\{x^{i,1},\dots,x^{i,m}\}\cup\{y^{i,2},\dots,y^{i,m}\}\cup\{z^{i,1},\dots,z^{i,m-1}\}$,
\item $\succ_{d_i^0}:~r^{i,1}\succ_{d_i^0}r^{i,2}\succ_{d_i^0}r^{i,3}$ for $i=1,\dots,n$,
\item $\succ_{d_i^1}:~x^{i,1}\succ_{d_i^1}z^{i,1}$ for $i=1,\dots,n$,
\item $\succ_{d_i^j}:~y^{i,j}\succ_{d_i^j}x^{i,j}\succ_{d_i^j}z^{i,j}$ for $i=1,\dots,n$ and $j=2,\dots,m-1$,
\item $\succ_{d_i^m}:~y^{i,m}\succ_{d_i^m}x^{i,m}$ for $i=1,\dots,n$,
\end{itemize}}
where\\
\scalebox{0.85}{\parbox{\linewidth}{
\begin{align*}
\arraycolsep=2pt
\begin{array}{lllll}
r^{i,1}_D=d_i^0,& r^{i,1}_H=h^+,      & u(r^{i,1})=\frac{a_i}{t},    & s(r^{i,1})=\frac{a_i}{t},  &(\substack{i=1,\dots,n}),\\
r^{i,2}_D=d_i^0,& r^{i,2}_H=h^-,      & u(r^{i,2})=\frac{a_i}{a-t},  & s(r^{i,2})=\frac{a_i}{a-t},&(\substack{i=1,\dots,n}),\\
r^{i,3}_D=d_i^0,& r^{i,3}_H=h_i^0     & u(r^{i,3})=1,                & s(r^{i,3})=1,              &(\substack{i=1,\dots,n}),\\
x^{i,j}_D=d_i^j,& x^{i,j}_H=h_i^0,    & u(x^{i,j})=\frac{2\alpha}{m},& s(x^{i,j})=\frac{1}{m},    &(\substack{i=1,\dots,n,\\j=1,\dots,m}),\\
y^{i,j}_D=d_i^j,& y^{i,j}_H=h_i^{j-1},& u(y^{i,j})=1,                & s(y^{i,j})=1               &(\substack{i=1,\dots,n,\\j=2,\dots,m}),\\
z^{i,j}_D=d_i^j,& z^{i,j}_H=h_i^{j},  & u(z^{i,j})=2\alpha,          & s(z^{i,j})=1               &(\substack{i=1,\dots,n,\\j=1,\dots,m-1}).
\end{array}
\end{align*}}}
We claim that
there exists a $1$-stable matching if the subset sum instance is yes-instance and
there exists no $\alpha$-stable matching otherwise.

First, suppose that the subset sum instance is yes-instance.
Let \(I^*\subseteq \{1,\dots,n\}\) such that $\sum_{i\in I}a_i=t$. 
Then, we observe that 
\begin{align*}
  X'\coloneqq\bigcup_{i\in I^*}\{r^{i,1}\}\cup\bigcup_{i\not\in I^*}\{r^{i,2}\}\cup\bigcup_{i=1}^n\{x^{i,1},y^{i,2},\dots,y^{i,m}\}
\end{align*}
is a $1$-stable matching.
Since every doctor $d_i^j$ ($i=1,\dots,n,~j=1,\dots,m$) matches with the best contract,
it is sufficient to consider coalitions that contain some $r^i$.
Let $X''=\{r^i\mid i\in I\}$ for a set $I\subseteq\{1,\dots,n\}$.
Then, $u(X'')=\sum_{i\in I}a_i\le t-1$ and hence $X''$ is not a $1$-blocking coalition.
Thus, $X'$ is a $1$-stable matching.

Next, we claim that 
$X'$ is not $\alpha$-stable if a matching $X'$ satisfies $X'\cap\{r^{i^*,1},r^{i^*,2}\}=\emptyset$ for some $i^*$.
Note that, if the subset sum instance is no-instance, there exists such an index $i^*$ for any matching $X'$.
Let us assume that $X'$ is an $\alpha$-stable matching.
We consider the following two cases.

\noindent\textbf{Case 1: $r^{i^*,3}\in X'$.}
In this case, $d_{i^*}^1$ is assigned to $h_{i^*}^1$, i.e., $z^{i^*,1}\in X'$ and $y^{i^*,2}\not\in X'$,
because $u(z^{i^*,1})=2\alpha\cdot u(y^{i^*,2})>\alpha\cdot u(y^{i^*,2})$.
By applying induction, we obtain that $z^{i^*,1},\dots,z^{i^*,m-1}\in X'$ and $y^{i^*,2},\dots,y^{i^*,m}\not\in X'$.
Thus, $\{x^{i^*,1},\dots,x^{i^*,m}\}$ is an $\alpha$-blocking coalition for $h_i^0$, a contradiction.

\noindent\textbf{Case 2: $r^{i^*,3}\not\in X'$.}
In this case, each hospital $h_{i^*}^{j-1}$ $(j=2,\dots,m)$
is assigned a contract since otherwise $\{y_{i^*}^j\}$ is an $\alpha$-blocking coalition for $h_{i^*}^{j-1}$.
Thus, we have $|X'\cap\{x^{i^*,1},\dots,x^{i^*,m}\}|\le 1$.
This implies that $u(X_{h_{i^*}^0}')\le 2\alpha/m<1/\alpha$ and hence $\{r^{i^*,3}\}$ is an $\alpha$-blocking coalition for $h_{i^*}^0$.
This contradicts our assumption.

Thus, there exists no $\alpha$-stable matching if the subset sum instance is no-instance.
\end{proof}

\newtheorem*{thm:knapsack_lower}{Theorem~\ref{thm:knapsack_lower}}
\begin{thm:knapsack_lower}
For any positive reals $1/2<\smax<1$ and $\epsilon$, 
there exists a market such that
$s(x)\le\smax$ for all $x$ and
no $\left(\frac{1}{1-\smax}-\epsilon\right)$-stable matching exists.
\end{thm:knapsack_lower}
\begin{proof}[Proof of Theorem~\ref{thm:knapsack_lower}]
Let $\delta$ be a sufficiently small positive real and $\alpha=1/(1-\smax)$.
Let $n=\lfloor 1/\delta^3\rfloor$, $m=\lceil 1/\smax\rceil$, and $l=\lfloor(1-\smax)/\delta\rfloor+1$.

Consider a market with $n+1$ doctors $D=\{d^*,d_1,\dots,d_{n}\}$ and $n-l+1$ hospitals $H=\{h^*,h_{l+1},\dots,h_n\}$.
The set of contracts is $X=\{x^*,x^1,\dots,x^{n}\}\cup\{y^{i,j}\mid i=1,\dots,n,~j=\max\{i,l+1\},\dots,n\}$
where\\
\scalebox{0.85}{\parbox{\linewidth}{
\begin{align*}
\arraycolsep=2pt
\begin{array}{lllll}
x^*_D=d^*,    & x^*_H=h^*,    & u(x)=\smax,               & s(x)=\smax,&\\
x^i_D=d_i,    & x^i_H=h^*,    & u(x^i)=i\cdot\delta^3,    & s(x^i)=\delta &(i=1,\dots,n),\\
y^{i,j}_D=d_i,& y^{i,j}_H=h_j,& u(y^{i,j})=\alpha^{-i},   & s(y^{i,j})=\smax&(\substack{i=1,\dots,n,\\j=\max\{i,l+1\},\dots,n}),\\
\end{array}
\end{align*}}}
The doctors' preferences are given as follows:\\
\scalebox{0.9}{\parbox{\linewidth}{
\begin{align*}
\succ_{d_0}:&~x^*,\\
\succ_{d_i}:&~x^i\succ_{d_i}y^{i,l+1}\succ_{d_i}\dots\succ_{d_i}y^{i,n} \quad(i=1,\dots,l)\\
\succ_{d_i}:&~y^{i,i}\succ_{d_i}x^i\succ_{d_i}y^{i,i+1}\succ_{d_i}\dots\succ_{d_i}y^{i,n} \quad(i=l+1,\dots,n).
\end{align*}}}

We claim that there exists no $\left(\alpha-\epsilon\right)$-stable matching in this market by contradiction.
Let us assume that $X'$ is an $\left(\alpha-\epsilon\right)$-stable matching.
Here, each hospital $h_i$ $(i=l+1,\dots,n)$ is assigned a contract since otherwise $\{y^{i,i}\}$ is an $(\alpha-\epsilon)$-blocking coalition for $h_i$.

First, let us consider a case where $x^*\in X'$.
In this case, at least one of doctors in $\{d_1,\dots,d_l\}$ is not assigned to $h^*$, i.e., $\{x^1,\dots,x^l\}\not\subseteq X'$, because $s(x^*)+\sum_{i=1}^l s(x^i)>1$.
Let $i_1$ be the smallest index $i$ such that $x^i\not\in X'$.
Then, we have $y^{i_1,l+1}\in X'$ and $y^{l+1,l+1}\not\in X'$ since otherwise $\{y^{i_1,l+1}\}$ is an $\alpha$-blocking coalition for $h_{l+1}$.
Similarly, at least one of doctors in $\{d_1,\dots,d_{l+1}\}\setminus\{d_{i_1}\}$ is not assigned to $h^*$.
Let $i_2$ be the second smallest index $i$ such that $x^i\not\in X'$.
Then, we have $y^{i_2,l+2}\in X'$ and $y^{l+2,l+2}\not\in X'$ since otherwise $\{y^{i_2,l+1}\}$ is an $\alpha$-blocking coalition for $h_{l+2}$.
By applying induction, we obtain that $y^{i,i}\not\in X'$ for $i=l+1,\dots,n$.
Now we show that $X^1\coloneqq\{x^{n-\lfloor 1/\delta\rfloor+1},\dots,x^n\}$ is an $\left(\alpha-\epsilon\right)$-blocking coalition for $h^*$.
By the above discussion, $x\succ_{x_D} x'$ holds for any $x\in X^1\setminus X'$ and $x'\in X_{x_D}'$.
Hence, it is sufficient to prove that $u(X^1)>\left(\alpha-\epsilon\right)\cdot u(X_{h^*}')$.
A simple computation gives that
\begin{align*}
u(X_{h^*}')
&\le \smax+\sum_{i=n-l+2}^n i\cdot \delta^3\\
&\le \smax+(l-1)n\cdot \delta^3
\le 1+(1-\smax)(1/\delta),
\end{align*}
and
\begin{align*}
u(X^1)
&=\sum_{i=n-\lfloor 1/\delta\rfloor+1}^n i\cdot \delta^3\\
&\ge (n-\lfloor 1/\delta\rfloor+1)\cdot \lfloor 1/\delta\rfloor\cdot \delta^3
\ge (1-\delta^2)\lfloor 1/\delta\rfloor
\end{align*}
since $s(X^1)=\delta\cdot\lfloor 1/\delta\rfloor\le 1$.
Thus, we obtain
\begin{align*}
\frac{u(X^1)}{u(X_{h^*}')}\ge \frac{(1-\delta^2)\lfloor 1/\delta\rfloor}{1+(1-\smax)(1/\delta)}\to \frac{1}{1-\smax}=\alpha
\end{align*}
as $\delta$ goes to $0$.
Hence, $X^1$ is an $\left(\alpha-\epsilon\right)$-blocking coalition for $h^*$
since $\delta$ is sufficiently small.

Next, let us consider the other case, i.e. $x^*\not\in X'$.
Since each hospital $h_i$ $(i=l+1,\dots,n)$ is assigned a contract, we have $|X'\cap\{x^1,\dots,x^n\}|\le l$.
Also, $|X'\cap\{x^1,\dots,x^n\}|\ge 1$ since otherwise $\{x^*\}$ is a $(\alpha-\epsilon)$-blocking coalition.
Let $i^*$ be the largest index $i$ such that $x^i\in X'$.
Then we have $u(X_{h^*}')\le l\cdot i^*\cdot \delta^3\le i^*\cdot((1-\smax)\delta^2+\delta^3)$.

For the case $i^*\le \smax/\delta^2$, let $X^2=\{x^*\}$.
Then we get
\begin{align*}
\frac{u(X^2)}{u(X_{h^*}')}
\ge \frac{\smax}{i^*\cdot((1-\smax)\delta^2+\delta^3)}
\ge \frac{1}{(1-\smax)+\delta}
\to\alpha
\end{align*}
as $\delta$ goes to $0$.
Hence, $X^2$ is an $\left(\alpha-\epsilon\right)$-blocking coalition for $h^*$ when $\delta$ is sufficiently small.

For the case $i^*> \smax/\delta^2$, let $X^3=\{x^{i^*-\lfloor 1/\delta\rfloor+1},\dots,x^{i^*}\}$.
Here, $i^*>l$ holds since $\delta$ is sufficiently small.
We show that $X^3$ is an $\left(\alpha-\epsilon\right)$-blocking coalition for $h^*$ in this case.
Since $|X'\cap\{x^1,\dots,x^n\}|\le l$ and $x^{i^*}\in X'$, at least one of doctors in $\{d_1,\dots,d_l\}$ is not assigned to $h^*$, i.e., $\{x^1,\dots,x^l\}\not\subseteq X'$.
Analysis similar to that in the first case shows that $y^{i,i}\not\in X'$ for $i=l+1,\dots,i^*$.
Then, we have $x\succ_{x_D} x'$ holds for any $x\in X^3\setminus X'$ and $x'\in X_{x_D}'$.
Also we have
\begin{align*}
\frac{u(X^3)}{u(X_{h^*}')}
&\ge \frac{\sum_{i=i^*-\lfloor 1/\delta\rfloor+1}^{i^*}i\cdot \delta^3}{i^*\cdot((1-\smax)\delta^2+\delta^3)}
\ge \frac{(i^*-\lfloor 1/\delta\rfloor)\cdot\delta^2}{i^*\cdot((1-\smax)\delta^2+\delta^3)}\\
&\ge \frac{\delta^2}{(1-\smax)\delta^2+\delta^3}-\frac{(\lfloor 1/\delta\rfloor)\cdot\delta^2}{(\smax/\delta^2)\cdot((1-\smax)\delta^2+\delta^3)}\\
&\to\alpha
\end{align*}
as $\delta$ goes to $0$.
Hence, $X^3$ is an $\left(\alpha-\epsilon\right)$-blocking coalition for $h^*$ when $\delta$ is sufficiently small.
\end{proof}

\begin{proposition}\label{prop:no-strategy-proof}
For $1\ge \smax>1/2$,
there exists a market such that 
the generalized DA with the choice functions defined in Algorithm~\ref{alg:knapsack} 
is not strategy-proof for doctors.
\end{proposition}
\begin{proof}
Let us consider the following market:
\begin{itemize}
\item $D=\{d_1,d_2,d_3\}$, $H=\{h_1,h_2\}$, $X=\{x^{1,1},x^{1,2},x^{2,1},x^{2,2},x^{3,1},x^{3,2}\}$,
\item $x^{i,j}_D=d_i$ and $x^{i,j}_H=h_j$ $(i=1,2,3,~j=1,2)$,
\item $s(x^{1,1})=s(x^{3,1})=1/2$, $s(x^{1,2})=s(x^{2,1})=s(x^{2,2})=s(x^{3,2})=\smax$,
\item $u(x^{1,1})=1,~u(x^{1,2})=4,~u(x^{2,1})=4,~u(x^{2,2})=1,~u(x^{3,1})=2,~u(x^{3,2})=2$,
\item $x^{1,1}\succ_{d_1}x^{1,2}$, $x^{2,2}\succ_{d_2}x^{2,1}$, $x^{3,2}\succ_{d_3}x^{3,1}$.
\end{itemize}

For this market, the generalized DA outputs a matching $\{x^{1,2},x^{3,1}\}$.
However, if $d_3$ misreports her preference as $x^{2,1}\succ_{d_3}'x^{2,2}$,
then the generalized DA gives a matching $\{x^{1,2},x^{2,1},x^{3,1}\}$.
Thus, the mechanism is not strategy-proof for doctors since $d_2$ has an incentive to misreport.
\end{proof}

\begin{theorem}\label{thm:porp_hard}
  For any sciently small $\epsilon>0$,
  distinguishing whether a given proportional market
  has $\left(\frac{1+\sqrt{5}}{2}-\epsilon\right)$-stable matching or not is NP-hard.
\end{theorem}
\begin{proof}
Recall that $\phi=\frac{1+\sqrt{5}}{2}$.

We prove the NP-hardness by reducing the subset sum problem.
Let $a_1,\dots,a_n$ and $t$ be an instance of subset sum problem and let $a=\sum_{j=1}^n a_j$.
Without loss of generality, we assume that $a>t>0$.

Consider a market with $n+3$ doctors $D=\{d_1,\dots,d_{n+2}\}$ and four hospitals $H=\{h_1,\dots,h_4\}$.
The set of constracts is $X=\bigcup_{i=1}^n\{x^{i,1},x^{i,2},x^{i,3}\}\cup\{y,\hat{y},z,\hat{z}\}$ where\\
\scalebox{0.9}{\parbox{\linewidth}{
\begin{align*}
\arraycolsep=3pt
  \begin{array}{lllll}
    x^{i,1}_D=d_i,    & x_H=h_1,      & s(x^{i,1})=a_i/t,&(i=1,\dots,n),\\
    x^{i,2}_D=d_i,    & x_H=h_2,      & s(x^{i,2})=a_i/(a-t),&(i=1,\dots,n),\\
    x^{i,3}_D=d_i,    & x_H=h_3,      & s(x^{i,3})=1/\phi+\epsilon/2,&(i=1,\dots,n),\\
    y_D=d_{n+1},      & y_H=h_3,      & s(y)=(1/\phi)^2,&\\
    \hat{y}_D=d_{n+1},& \hat{y}_H=h_4,& s(\hat{y})=1,&\\
    z_D=d_{n+2},      & z_H=h_3,      & s(z)=1/\phi,&\\
    \hat{z}_D=d_{n+2},& \hat{z}_H=h_4,& s(\hat{z})=\epsilon.&
  \end{array}
\end{align*}}}
The doctors' preferences are given as follows: 
\begin{align*}
  \succ_{d_{i}}:&~x^{i,1}\succ_{d_i}x^{i,2}\succ_{d_i}x^{i,3},&(i=1,\dots,n),\\
  \succ_{d_{n+1}}:&~y\succ_{d_{n+1}}\hat{y},\\
  \succ_{d_{n+2}}:&~\hat{z}\succ_{d_{n+2}}z.
\end{align*}
We claim that
there exists a $1$-stable matching if the subset sum instance is yes-instance and
there exists no $(\phi-\epsilon)$-stable matching otherwise.

Suppose that the subset sum instance is yes-instance.
Let \(I^*\subseteq \{1,\dots,n\}\) such that $\sum_{i\in I}a_i=t$. 
Then, we can observe that
\begin{align*}
  X''\coloneqq \{x^{i,1}\mid i\in I^*\}\cup\{x^{i,2}\mid i\not\in I^*\}\cup\{y,\hat{z}\}
\end{align*}
is a $1$-stable matching.

Conversely, suppose that the subset sum instance is no-instance.
We claim that there exists no $(\phi-\epsilon)$-stable matching by contradiciton.
Let us assume that $X''$ is an $(\phi-\epsilon)$-stable matching.
Since the subset sum instance is no-instance,
there exists an index $i^*\in\{1,\dots,n\}$ such that $X''\cap\{x^{i^*,1},x^{i^*,2}\}=\emptyset$.
Then, by a similar discussion to the proof of Theorem~\ref{thm:prop_nonexist}, $X''$ is not $(\phi-\epsilon)$-stable,
which is a contradiction.
\end{proof}

\end{document}